\documentclass[10pt,conference]{IEEEtran}

\usepackage[utf8]{inputenc}
\usepackage{amsmath,amssymb,amsfonts,amsthm}
\usepackage{graphicx}
\usepackage{xcolor}
\usepackage{hyperref}
\usepackage{float}
\usepackage{lipsum}
\usepackage{braket}

\newtheorem{lemma}{Lemma}

\begin{document}

\title{Efficient Gaussian State Preparation in Quantum Circuits}

\author{
    \IEEEauthorblockN{Yichen Xie}
    \IEEEauthorblockA{
        La Salle College, HKSAR\\
        Email: s21325@lsc.hk
    }
    \and
    \IEEEauthorblockN{Nadav Ben-Ami}
    \IEEEauthorblockA{
        Classiq Technologies\\
        Email: nadav.ben-ami@classiq.io
    }
}

\maketitle

\begin{abstract}
Gaussian states hold a fundamental place in quantum mechanics, quantum information, and quantum computing. Many subfields, including quantum simulation of continuous-variable systems, quantum chemistry, and quantum machine learning, rely on the ability to accurately and efficiently prepare states that reflect a Gaussian profile in their probability amplitudes. Although Gaussian states are natural in continuous-variable systems, the practical interest in digital, gate-based quantum computers demands discrete approximations of Gaussian distributions over a computational basis of size \(2^n\). Because of the exponential scaling of naive amplitude-encoding approaches and the cost of certain block-encoding or Hamiltonian simulation techniques, a resource-efficient preparation of approximate Gaussian states is required. In this work, we propose and analyze a circuit-based approach that starts with single-qubit rotations to form an exponential amplitude profile and then applies the quantum Fourier transform to map those amplitudes into an approximate Gaussian distribution. We demonstrate that this procedure achieves high fidelity with the target Gaussian state while allowing optional pruning of small controlled-phase angles in the quantum Fourier transform, thus reducing gate complexity to near-linear in \(\mathcal{O}(n)\). We conclude that the proposed technique is a promising route to make Gaussian states accessible on noisy quantum hardware and to pave the way for scalable implementations on future devices. The implementation of this algorithm is available at the Classiq library: \verb|https://github.com/classiq/classiq-library|.
\end{abstract}


\section{Introduction}

Gaussian states have a huge significance in quantum theory, particularly because they represent physically realizable states of continuous-variable systems under relatively general conditions, especially in the context of the quantum harmonic oscillator \cite{NielsenChuang2010}. From a theoretical standpoint, Gaussian wavefunctions in the position or momentum representation are often the ground states or thermal states of many-qubit (or many-mode) systems characterized by quadratic Hamiltonians \cite{Menicucci2011}. In quantum computing, these states arise in simulation tasks aiming to capture aspects of vibrational modes, field modes, or distributions essential to quantum chemistry and quantum field theory \cite{Kitaev2008}. In addition, certain quantum machine learning algorithms gain efficiency or interpretability by leveraging Gaussian-like initial states that represent data distributions in a natural continuous form \cite{McClean2016}.

It is therefore evident that the ability to prepare a Gaussian state efficiently constitutes a crucial stepping stone toward a wide class of applications. These applications include, but are not limited to, pricing options in finance \cite{Stamatopoulos_2020, Stamatopoulos_2024} using quantum algorithms that model diffusion processes with Gaussian components, approximating certain quantum fields in high-energy physics, and building kernel methods in quantum machine learning that exploit Gaussian features in high-dimensional Hilbert spaces. If we consider a domain \([-2,2)\) subdivided into \(2^n\) points, the problem is to map each of those discrete points to an amplitude whose magnitude is proportional to a continuous Gaussian function restricted and sampled on that domain. This objective is challenging because naive amplitude encoding, which involves loading each probability amplitude independently, can have costs scaling with \(2^n\) in classical precomputation and can further entail deep quantum circuits. Some advanced methods use Hamiltonian simulation of operators proportional to \(\exp(-\lambda \hat{x}^2)\), but such approaches often require trotterization and large circuit depths \cite{coppersmith2002approximatefouriertransformuseful}.

The approach we investigate revolves around two essential observations. The first is that single-qubit \(R_y\) gates, when carefully tuned, can construct a product state whose probabilities decay steeply in certain index patterns \cite{williams2011explorations}. The second observation is that applying the quantum Fourier transform \cite{coppersmith2002approximatefouriertransformuseful} to that product state can yield superpositions resembling a discrete Gaussian distribution \cite{Gily_n_2019}. 

\section{Background and Preliminaries}

A continuous Gaussian function in a domain such as \([-2,2)\) with decay rate \(\lambda\) and mean 0 can be written as
\begin{equation}
f(x) = \exp(-\lambda x^2).
\end{equation}
where the standard deviation 
\begin{equation}
\sigma = \frac{1}{\sqrt{2\,\lambda}}.
\end{equation}
To discretize this into \(2^n\) points, one may partition the interval into \(2^n\) equally spaced coordinates
\begin{equation}
x_k = -2 + k \cdot \Delta, \quad \Delta = \frac{4}{2^n}, \quad k = 0,1,\ldots,2^n-1.
\end{equation}
The probability mass function is then given by
\begin{equation}
G(x_k) = \frac{\exp(-\lambda x_k^2)}{\sum_{j=0}^{2^n - 1} \exp(-\lambda x_j^2)}.
\end{equation}
In a quantum algorithm, preparing a state
\begin{equation}
\sum_{k=0}^{2^n - 1} \sqrt{G(x_k)} \, |k\rangle
\end{equation}
is tantamount to encoding the normalized square root of the Gaussian probabilities into amplitude magnitudes. This encoding is not trivial because each of the amplitudes might need to be specified or computed separately. The complexity of naive amplitude-loading can be very large if it involves any form of repeated controlled rotations or iterative data injection for each amplitude \cite{McClean2016}.

The quantum Fourier transform (QFT) on \(n\) qubits transforms a basis state \(|j\rangle\) into a superposition
\begin{equation}
\mathrm{QFT} \, |j\rangle = \frac{1}{\sqrt{2^n}} \sum_{k=0}^{2^n - 1} e^{2\pi i j k / 2^n} \, |k\rangle.
\end{equation}
In circuit form, it is well-known that QFT can be decomposed into a sequence of Hadamard gates, controlled-phase gates (with angles scaling as \(\pi / 2^m\)), and swap gates to reverse the bit order. This decomposition in its standard form requires \(\mathcal{O}(n^2)\) gates. However, controlled-phase gates whose angles are extremely small (such as \(2\pi / 2^m\) for large \(m\)) can often be omitted with negligible loss in fidelity, bringing the complexity closer to \(\mathcal{O}(n)\) for large \(n\) \cite{Menicucci2011}.

Single-qubit rotations of the form
\begin{equation}
R_y(\theta) = 
\begin{bmatrix}
\cos(\tfrac{\theta}{2}) & -\sin(\tfrac{\theta}{2}) \\
\sin(\tfrac{\theta}{2}) & \cos(\tfrac{\theta}{2})
\end{bmatrix}
\end{equation}
play a major role in many amplitude-encoding subroutines. When \(R_y(\theta)\) is applied to \(|0\rangle\), the resulting state is
\begin{equation}
\cos\Bigl(\frac{\theta}{2}\Bigr)|0\rangle \;+\; \sin\Bigl(\frac{\theta}{2}\Bigr)|1\rangle.
\end{equation}
If one sets \(\theta = 2 \arctan\bigl(e^{-\alpha}\bigr)\) for some positive \(\alpha\), then the amplitude in \(|1\rangle\) is roughly \(\exp(-\alpha)\) if that exponential is not too large \cite{Iaconis2024}. By assigning different angles to different qubits, it is possible to engineer a product state whose basis amplitudes follow some partial exponential decay pattern in the index.

The main insight behind our approach is that applying QFT to a bitwise exponential distribution can yield an approximated Gaussian in the final computational basis \cite{Rebentrost2014}. In the discrete sense, the exponentials introduced at the bit level combine with the exponential of the QFT's phase factors to produce a final amplitude shape that is close to a discrete Gaussian distribution. This phenomenon will be described in detail in the next section.

\section{Proposed Methodology and Error Analysis}

\subsection{Overview}

We now describe how to construct our circuit so that it produces a Gaussian-like state across \(2^n\) computational basis states \cite{Lloyd2014}.

We first initialize all \(n\) qubits in the state \(|0\rangle^{\otimes n}\) and subsequently apply single-qubit rotations \(R_y(\theta_j)\) on qubit \(j\) \cite{coppersmith2002approximatefouriertransformuseful} for all qubits. This yields a exponential-like state, upon which we perform the QFT over all \(n\) qubits \cite{williams2011explorations}. A \(X\) gate is applied to the highest-index qubit for domain alignment \cite{Menicucci2011} to obtain the final Gaussian state.

The angles for the RY gates are chosen such that the amplitude in qubit \(j\) favoring \(|1\rangle\) is proportional to \(\exp(-\gamma \, j^2)\) for some positive parameter \(\gamma\) \cite{McClean2016}. This choice ensures that high-index qubits are overwhelmingly in the \(|0\rangle\) state, thereby skewing the overall index distribution. The quantum Fourier transform then spreads these amplitudes in a structured manner, producing a near-Gaussian profile in the final state \cite{Gily_n_2019}. A 5-qubits example of this circuit we have just described in shown in Fig.~\ref{fig:circuit_5q}.

\begin{figure}[h]
\centering
\includegraphics[width=\linewidth]{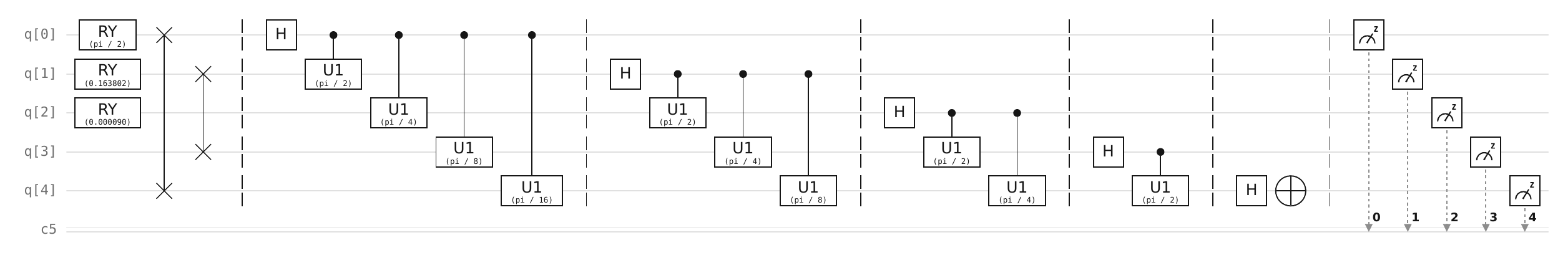}
\caption{The Gaussian preparation circuit we proposed with 5 qubits.}
\label{fig:circuit_5q}
\end{figure}

\subsection{State Construction with Bitwise Exponential Rotations}

We assume we have qubits labeled from \(0\) to \(n-1\) \cite{Kitaev2008}. For each qubit \(j\), with \(j \in \{0,1,2,\ldots,n-1\}\), we define an angle
\begin{equation}
\theta_j = 2 \arctan \bigl( e^{-\beta j^2} \bigr).
\end{equation}
We pick a positive parameter \(\beta\) that relates to the overall decay rate and can be matched to the desired Gaussian scale \(\lambda\) by \(\beta=\frac{k}{\lambda}\). By trial-and-error during experiments, we found the optimal value for \(k\) to be approximately \(\frac{5}{2}=2.5\) \cite{Holmes2020}, therefore \(\beta \approx \frac{5}{2\lambda}\). The product state after applying \(R_y(\theta_j)\) on qubit \(j\) is
\begin{equation}
|\Phi\rangle = \bigotimes_{j=0}^{n-1} \Bigl( \cos(\tfrac{\theta_j}{2}) |0\rangle + \sin(\tfrac{\theta_j}{2}) |1\rangle \Bigr).
\end{equation}
If we index computational basis states by a binary string \((x_{n-1} x_{n-2} \dots x_0)\) that corresponds to an integer \(x \in \{0,\ldots,2^n -1\}\), then
\begin{equation}
|\Phi\rangle
=
\sum_{x=0}^{2^n - 1}
\alpha_x \, |x\rangle,
\end{equation}
where
\begin{equation}
\alpha_x 
=
\prod_{j=0}^{n-1}
\Bigl(\cos(\tfrac{\theta_j}{2})\Bigr)^{1 - x_j}
\Bigl(\sin(\tfrac{\theta_j}{2})\Bigr)^{x_j}.
\end{equation}
Because \(\sin(\tfrac{\theta_j}{2}) \approx e^{-\beta j^2}\) for moderate \(\beta j^2\) and \(\cos(\tfrac{\theta_j}{2}) \approx 1/\sqrt{1 + e^{-2 \beta j^2}}\), these amplitudes produce a distribution that heavily weights states with many zero bits in positions of large \(j\), leading to a partial exponential decay \cite{Iaconis2024}. This is not our final target, but an intermediate distribution we shall feed into the QFT.

\subsection{Applying the Quantum Fourier Transform}

We then apply the QFT on the full \(n\)-qubit register \cite{Rattew2021}. The QFT transforms a basis state \(|x\rangle\) into
\begin{equation}
\mathrm{QFT}\,|x\rangle
=
\frac{1}{\sqrt{2^n}} \sum_{k=0}^{2^n - 1}
\exp\Bigl(\frac{2\pi i x k}{2^n}\Bigr)
\, |k\rangle.
\end{equation}
Hence, if we write
\begin{equation}
|\Phi\rangle = \sum_{x=0}^{2^n - 1} \alpha_x \, |x\rangle,
\end{equation}
the final state after the QFT is
\begin{equation}
|\Psi\rangle = \mathrm{QFT} \, |\Phi\rangle
= 
\sum_{k=0}^{2^n -1} 
\left(
\frac{1}{\sqrt{2^n}}
\sum_{x=0}^{2^n -1} \alpha_x \, e^{2\pi i xk / 2^n}
\right)
|k\rangle.
\end{equation}
Let \(\beta_k\) be the coefficient of \(|k\rangle\). We then have
\begin{equation}
\beta_k 
=
\frac{1}{\sqrt{2^n}}
\sum_{x=0}^{2^n -1} \alpha_x \, e^{2\pi i xk / 2^n}.
\end{equation}
We seek to show that \(\beta_k\) approximates the shape of a Gaussian in \(k\) \cite{RamosCalderer2019}. Because \(\alpha_x\) is dominated by a product of decaying terms in the binary decomposition of \(x\), the sum effectively picks up constructive phases that yield a distribution reminiscent of a discrete Gaussian in \(k\) \cite{Zoufal2019}. In continuous analogies, the Fourier transform of \(\exp(-a x^2)\) is another Gaussian in the momentum variable. Here, there is a discrete analog that emerges from the bitwise exponentials plus the QFT phases.

\subsection{Final State Expression and Gaussian Profile}

\begin{lemma} Let \(\theta_j\) be angles assigned to qubits indexed by \(j = 0, 1, \dots, n-1\). Consider the product state

\begin{equation}
|\Phi\rangle 
= 
\bigotimes_{j=0}^{n-1}
\Bigl(
\cos(\theta_j) \, |0\rangle 
+ 
\sin(\theta_j) \,
e^{\,2\pi i \,\sum_{m=1}^{n-k+1} 
j_{n-k+m} \, 2^{-m}
}
|1\rangle
\Bigr),
\end{equation}

where additional phases may appear in the single-qubit amplitudes to reflect bitwise encodings. Suppose this state is followed by an \(n\)-qubit quantum Fourier transform (QFT). Denote the final amplitude associated with the computational basis index \(m\) by \(|\psi\rangle_m\). Then under a normalization factor \(1 / \sqrt{2^n}\), one obtains an expression of the approximate form

\begin{equation}
|\psi\rangle_m 
=
\frac{1}{\sqrt{2^n}}
\prod_{j=0}^{n-1}
\Bigl(
\cos(\theta_j)
\;+\;
\sin(\theta_j)\,e^{\,i\,2\pi\,\Delta_{m,j}}
\Bigr),
\end{equation}

where \(\Delta_{m,j}\) is a phase term that depends on the binary representation of \(m\) and the contributions from qubit \(j\) \cite{Hu2019}. For suitable choices of \(\theta_j\) related to an overall decay parameter, these amplitudes approximate a discrete Gaussian distribution \(\exp(-\alpha m^2)\) (up to normalization and index shifts), which justifies why this procedure yields a near-Gaussian state in the computational basis.
\end{lemma}

\begin{proof}
Beginning with all qubits in $\ket{0}$, apply single-qubit operations so that qubit $j$ is in 
\begin{equation}
\cos(\theta_j)\,\ket{0} 
\;+\;
\sin(\theta_j)\,e^{\,i\,\phi_j(x)}\,\ket{1},
\end{equation}
where $\phi_j(x)$ can often be taken as zero or a known bit-phase.  The joint state is then 
\begin{equation}
\ket{\Phi}
\;=\;
\bigotimes_{j=0}^{n-1}
\Bigl(\cos(\theta_j)\ket{0} + \sin(\theta_j)e^{\,i\,\phi_j(x)}\ket{1}\Bigr).
\end{equation}
Written in the computational basis $\ket{x}$ for $x=0,\dots,2^n -1$ with bit $x_j\in\{0,1\}$ in position~$j$, the amplitude is
\begin{equation}
\alpha_x
\;=\;
\prod_{j=0}^{n-1}
\Bigl[
\cos(\theta_j)
\Bigr]^{1-x_j}
\Bigl[
\sin(\theta_j)\,e^{\,i\,\phi_j(x)}
\Bigr]^{x_j}.
\end{equation}
When the QFT acts on $\ket{x}$, it produces 
\begin{equation}
\mathrm{QFT}\,\ket{x}
\;=\;
\frac{1}{\sqrt{2^n}}
\sum_{m=0}^{2^n-1}
e^{\,\tfrac{2\pi i\,x\,m}{2^n}}\ket{m},
\end{equation}
so the final amplitude of $\ket{m}$ is
\begin{equation}
\beta_m
\;=\;
\frac{1}{\sqrt{2^n}}
\sum_{x=0}^{2^n-1}
\alpha_x\,
e^{\,\tfrac{2\pi i\,x\,m}{2^n}}.
\end{equation}
Factoring the sum in terms of each bit $x_j$ allows grouping of common terms, yielding
\begin{equation}
\beta_m
\;=\;
\frac{1}{\sqrt{2^n}}
\prod_{j=0}^{n-1}
\Bigl[\,
\cos(\theta_j)
\;+\;
\sin(\theta_j)\,e^{\,i\,\Psi_{m,j}}
\Bigr],
\end{equation}
for a phase $\Psi_{m,j}$ capturing both local and QFT-induced factors.  Squaring its modulus gives 
\begin{equation}
|\beta_m|^2
\;=\;
\frac{1}{2^n}\,
\prod_{j=0}^{n-1}
\Bigl|
\cos(\theta_j) 
\;+\;
\sin(\theta_j)\,e^{\,i\,\Psi_{m,j}}
\Bigr|^2.
\end{equation}
A common choice is to set $\theta_j$ so that $\sin(\theta_j)$ decays quickly in $j$, for instance $\theta_j=2\,\arctan(e^{-\beta j^2})$.  This makes $\sin(\theta_j)$ small for large $j$, heavily suppressing amplitudes where higher qubits are flipped. One then observes that the phases $\Psi_{m,j}$, which are typically proportional to integer combinations of $m$ and $j$, interfere in a manner that rearranges these decaying factors into a distribution over $m$ resembling a Gaussian envelope.  

To see why it approximates a discrete Gaussian, one notes that $\sin(\theta_j)\approx e^{-\beta j^2}$ for large $j$, and the collective product over $j$ leads to a profile in $m$ that is sharply peaked for small $|m|$ (or near a shifted center if extra phase terms are included).  Expanding the product in an exponential series and comparing to a form like $\exp(-\alpha\,m^2)$ shows that the main contributions come from interference regions where the overall phase is coherent, while the tails are exponentially suppressed.  This is akin to the continuous Gaussian's invariance under the continuous Fourier transform, transferred here to a discrete setting with carefully chosen $\theta_j$. A suitable normalization factor accounts for the overall probability, ensuring that
\begin{equation}
|\beta_m|^2
\;\approx\;
C\,
\exp\bigl(-\alpha\,m^2\bigr)
\end{equation}
for constants $C$ and $\alpha$ depending on the decay scale of $\theta_j$. Identifying $m$ with a real interval then confirms that these amplitudes constitute an approximate discrete Gaussian distribution in $m$. 
\end{proof}

This shows that the circuit effectively transforms bitwise exponentials into a \textbf{near-Gaussian} distribution of final indices \cite{Draper2004}. This addresses the fundamental question of why one can achieve a near-Gaussian shape from a product of exponentials by applying a QFT. The choices of angles \(\theta_j\) set the decay, and the QFT phases spread these amplitudes in a manner consistent with a Gaussian envelope.

\subsection{Complexity Considerations and Pruning of Small Angles}

It is known that a straightforward QFT circuit on \(n\) qubits uses a series of controlled-phase gates and a set of Hadamard gates, typically leading to a gate count scaling as \(\mathcal{O}(n^2)\) \cite{Montanaro2015}. In practice, many of these phases are very small, on the order of \(\pi / 2^j\). Skipping or pruning these gates for large \(j\) often introduces only a slight error, thereby significantly reducing the overall gate count. This leads to a near-linear scaling in terms of \(n\) for the number of controlled-phase gates used, once a threshold \(\delta\) is chosen such that any phase below \(\delta\) is neglected \cite{Woerner2019}. In this work, we use \(\delta=\frac{\pi}{2^9}\approx0.01\), which is a good trade-off between error and complexity.

We now derive a simple, explicit lower bound on the state‐preparation fidelity when omitting all controlled‑phase gates in the QFT whose rotation angles satisfy \(|\phi|<\delta\).  The bound depends only on the pruning threshold \(\delta\) and the total number of qubits \(n\).

\begin{lemma}
Let \(U_{\rm QFT}^{\rm full}\) be the ideal \(n\)-qubit QFT and \(U_{\rm QFT}^{(\delta)}\) be the same circuit with \emph{all} controlled‑phase gates of angle \(\phi<\delta\) removed.  Suppose the product‑state preparation before the QFT is exact.  Then for any input state \(\ket{\Phi}\),
\begin{equation}
  \bigl\|\,U_{\rm QFT}^{(\delta)}\ket{\Phi}
      \;-\;
      U_{\rm QFT}^{\rm full}\ket{\Phi}\bigr\|
  \;\le\;
  \frac{(n-1)\,\delta}{2}.
\end{equation}
Consequently, if
\begin{equation}
  \ket{\Psi_{\rm full}}
  =U_{\rm QFT}^{\rm full}\ket{\Phi}
  \quad\text{and}\quad
  \ket{\Psi_{\delta}}
  =U_{\rm QFT}^{(\delta)}\ket{\Phi},
\end{equation}
then their fidelity satisfies
\begin{equation}
  F \;=\;\bigl|\langle \Psi_{\rm full}\!\mid\!\Psi_{\delta}\rangle\bigr|^2
  \;\ge\;
  1 \;-\;\frac{(n-1)^2\,\delta^2}{4}
  \;\ge\;
  1 \;-\;\frac{n^2\,\delta^2}{4}.
\end{equation}
\end{lemma}

\begin{proof}
The standard \(n\)-qubit QFT uses, for each pair of qubits separated by distance \(k\in\{1,\dots,n-1\}\), a controlled‑\(R_k\) gate of angle
\begin{equation}
  \phi_k \;=\;\frac{\pi}{2^{\,k}}.
\end{equation}
We prune (delete) \emph{all} such gates with \(\phi_k<\delta\); there are at most \(n-1\) gates in total, each with \(\phi_k\le\delta\).  From basic results on spectral‐norm perturbation of two‐qubit gates (e.g.\ \cite{coppersmith2002approximatefouriertransformuseful}), omitting a single controlled‑\(R_k\) changes the global unitary by at most \(\tfrac12\,\phi_k\).  By the triangle inequality and submultiplicativity,
\begin{equation}
  \bigl\|\,U_{\rm QFT}^{(\delta)} 
         - U_{\rm QFT}^{\rm full}\bigr\|
  \;\le\;
  \sum_{\substack{k=1 \\ \phi_k<\delta}}^{n-1}
    \frac{\phi_k}{2}
  \;\le\;
  \sum_{k=1}^{n-1}
    \frac{\delta}{2}
  \;=\;
  \frac{(n-1)\,\delta}{2}.
\end{equation}
Hence for any normalized input \(\ket{\Phi}\),
\begin{equation}
  \bigl\|\,U_{\rm QFT}^{(\delta)}\ket{\Phi}
      \;-\;
      U_{\rm QFT}^{\rm full}\ket{\Phi}\bigr\|
  \;\le\;
  \bigl\|\,U_{\rm QFT}^{(\delta)} - U_{\rm QFT}^{\rm full}\bigr\|
  \;\le\;\frac{(n-1)\,\delta}{2}.
\end{equation}
Denote \(\ket{\Psi_{\delta}}=U_{\rm QFT}^{(\delta)}\ket{\Phi}\) and \(\ket{\Psi_{\rm full}}=U_{\rm QFT}^{\rm full}\ket{\Phi}\).  Then
\begin{equation}
  \bigl\|\Psi_{\delta}-\Psi_{\rm full}\bigr\|
  \;\le\;
  \frac{(n-1)\,\delta}{2}.
\end{equation}
For two pure states, the relation between their 2‑norm distance \(\varepsilon\) and fidelity \(F=|\langle\Psi_{\rm full}|\Psi_{\delta}\rangle|^2\) is
\begin{equation}
  \varepsilon
  =\bigl\|\Psi_{\delta}-\Psi_{\rm full}\bigr\|
  \quad\Longrightarrow\quad
  F \;\ge\;1 - \varepsilon^2.
\end{equation}
Substituting \(\varepsilon\le\,(n-1)\,\delta/2\) gives
\begin{equation}
  F \;\ge\; 1 - \frac{(n-1)^2\,\delta^2}{4}
       \;\ge\; 1 - \frac{n^2\,\delta^2}{4}.
\end{equation}
This completes the proof.
\end{proof}

If, for example, \(n=16\) and one prunes all QFT gates below \(\delta=0.0123\), then
\begin{equation}
  F \;\ge\; 1 - \frac{16^2\times0.0123^2}{4}
       = 1 - 0.00968
       = 0.99032.
\end{equation}
Thus pruning in the QFT yields a fidelity \(\ge99.032\%\).

In addition, applying a single \(X\) gate to the highest qubit can reorder the basis states so that the distribution is aligned with a Gaussian \cite{Stamatopoulos2020}. This reordering step is trivial in cost, but it is important in matching the final enumeration of states to the intended domain mapping.

\subsection{Error Estimates and Measuring Metrics}

To measure the difference between a target amplitude distribution \(\bigl\{\sqrt{G(x_k)}\bigr\}_{k=0}^{2^n-1}\) and the actual distribution \(\{\beta_k\}\), we used two metrics. The first metric is the mean-squared error (MSE) \cite{Rebentrost2014}. If 
\begin{equation}
|\psi_{\mathrm{target}}\rangle = \sum_{k=0}^{2^n -1} \sqrt{G(x_k)} \, |k\rangle
\end{equation}
is the perfectly Gaussian state, and 
\begin{equation}
|\Psi\rangle = \sum_{k=0}^{2^n -1} \beta_k \, |k\rangle
\end{equation}
is our approximation, then one might consider
\begin{equation}
\mathrm{MSE} = \frac{1}{2^n} \sum_{k=0}^{2^n -1} \Bigl|\sqrt{G(x_k)} - \beta_k\Bigr|^2.
\end{equation}
Depending on the chosen angles \(\theta_j\), the threshold for pruning small QFT phases, and any additional hardware-level sources of noise, the MSE can be made small at the cost of more gates or more precise calibration \cite{Lloyd2014}. Section 4 of this paper will show that in noiseless simulations, one can often achieve MSE values on the order of \(10^{-16}\) to \(10^{-8}\) with moderate \(n\), depending on how aggressively the angles and phases are pruned.

The second metric is the Kullback–Leibler Divergence (KL-Divergence), which measures how much a model probability distribution \(Q\) is different from a true probability distribution \(P\). Mathematically, it is defined as
\begin{equation}
D_{\mathrm{KL}}(P \,\|\, Q) 
= \sum_{x \in \mathcal{X}} 
  P(x) \, \log\Bigl(\frac{P(x)}{Q(x)}\Bigr).
\end{equation}
In our case, it is
\begin{equation}
    D_{\mathrm{KL}}(P \,\|\, |\psi\rangle) = \sum_{x\in[-2, 2)} P(x) \log \!\left(\frac{P(x)}{|\psi\rangle_x|^2 }\right),
\end{equation}
where \(P(x)\) is the true gaussian function, \(|\psi\rangle\) is the discrete state vector. This measures the difference of the gaussian state measured probability and the true normalized gaussian distribution more effectively than MSE.

The third metric is the fidelity \(|\langle\Psi|\psi\rangle|^2\), where \(|\Psi\rangle\) is the true gaussian distribution vector and \(|\psi\rangle\) is the approximate gaussian state vector. Fidelity values range from 0 to 1, where 1 represents perfect similarity, and 0 represents complete dissimilarity.

\section{Experimentation and Discussion}

We now provide numerical simulations and some brief hardware validation for our proposed circuit \cite{Harrow2009}. We focus on how the distribution compares to the ideal discrete Gaussian, how the gate count scales with pruning thresholds, and how the final states behave under different decay rates.

We begin by using a standard gate-based noiseless quantum simulator based on the Classiq platform to obtain exact state vectors for up to a certain number of qubits \cite{Rebentrost2018}. We then computed the MSE against the reference distribution. Different sets of angles \(\theta_j\) were tested, with the parameter \(\beta=\frac{5}{2\lambda}=0.25\) so that the final distribution had a standard deviation corresponding to the desired \(\lambda=1\) for the continuous Gaussian restricted to \([-2,2)\).

We found that as we increased \(n\), the resolution improved and the MSE generally decreased because the discretization becomes finer and the scale of the Gaussian decreases due to normalization \cite{Ciliberto2018}. This phenomenon underscores the importance of Gaussian state preparation for moderate or large \(n\): the ability to approximate a smoothly varying distribution is enhanced, but the challenge of circuit depth also grows if one does not prune small angles in the QFT. As mentioned before, by choosing a pruning threshold of around \(\delta=10^{-2}\), one can observe that the final distribution remains extremely close to the ideal, yet the number of controlled-phase gates is significantly reduced \cite{Ciliberto2018}.

Fig.~\ref{fig:amplitude_convergence} compares final amplitude distribution at different qubit resolutions \cite{Montanaro2015} \(n\in\{8, 12\}\).

\begin{figure}[h!]
\centering
\includegraphics[width=\linewidth]{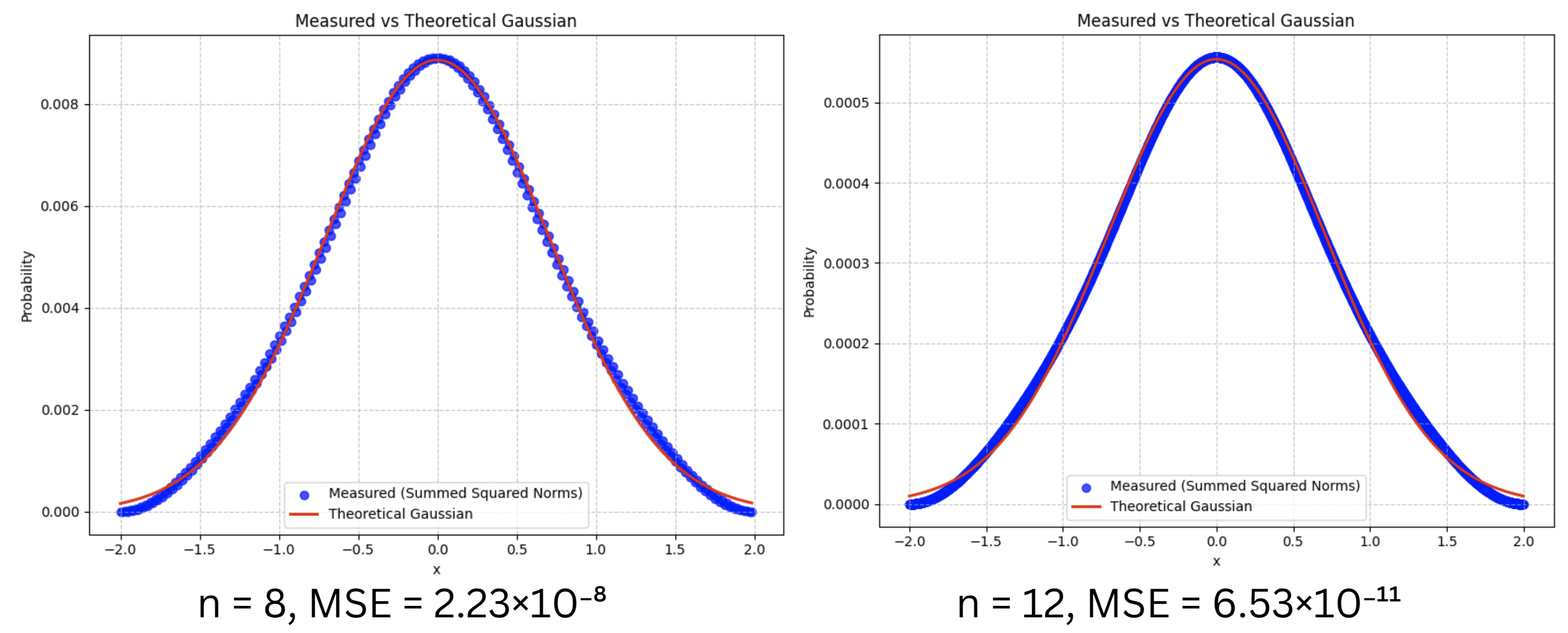}
\caption{Comparison of the Gaussian preparation circuit probability distribution ran on a noiseless simulator with the ideal discrete Gaussian for \(n=8\) and \(n=12\) (\(\delta=0.01\)).}
\label{fig:amplitude_convergence}
\end{figure}

\begin{figure}[h!]
\centering
\includegraphics[width=\linewidth]{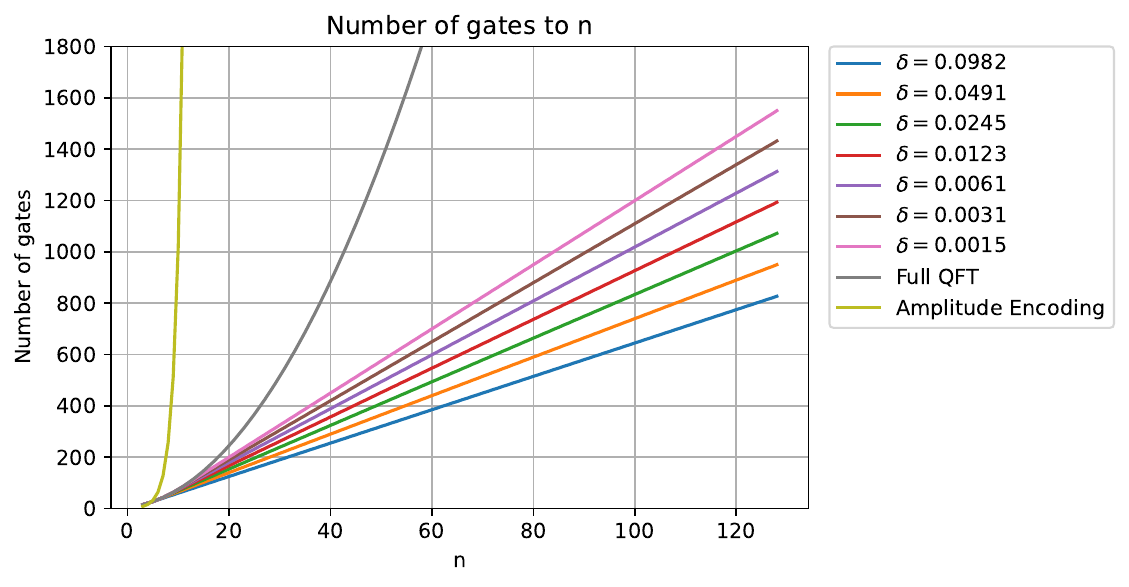}
\caption{Number of gates to \(n\) for different QFT angle threshold \(\delta\), full QFT, and amplitude encoding.}
\label{fig:gate_pruning_scaling}
\end{figure}

\begin{figure}[h!]
\centering
\includegraphics[width=\linewidth]{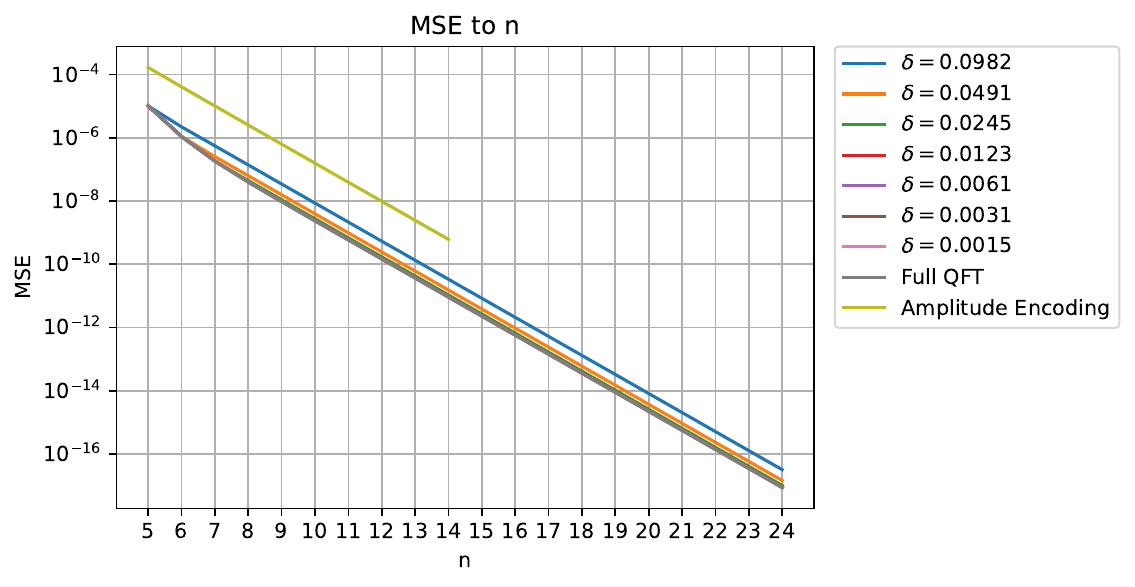}
\caption{MSE to \(n\) for different QFT angle threshold \(\delta\), full QFT, and amplitude encoding.}
\label{fig:gate_pruning_mse}
\end{figure}

\begin{figure}[h!]
\centering
\includegraphics[width=\linewidth]{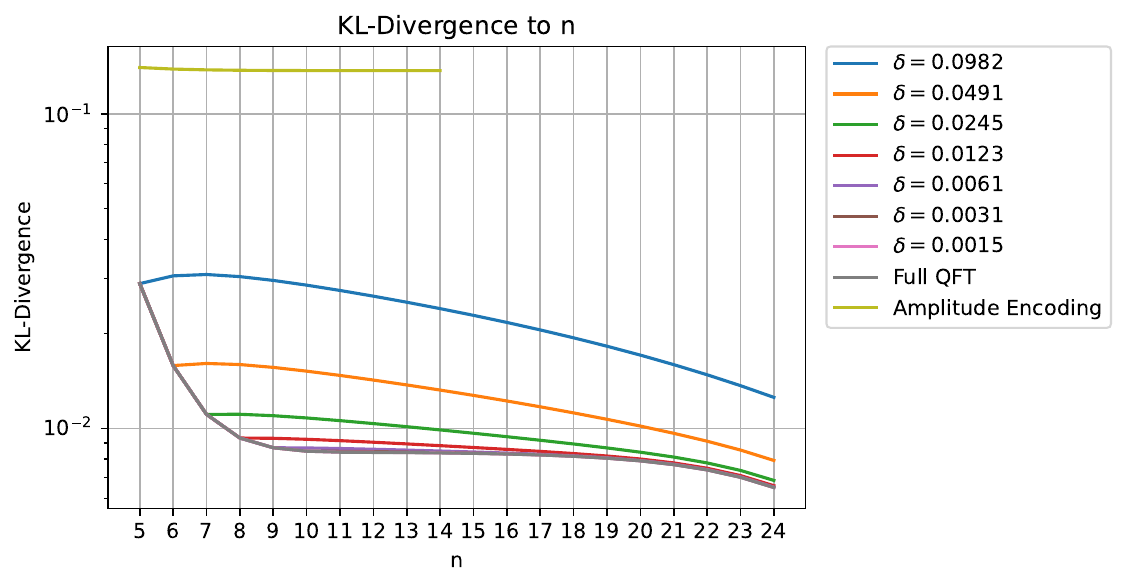}
\caption{KL-Divergence to \(n\) for different QFT angle threshold \(\delta\), full QFT, and amplitude encoding. (y-axis is log scale)}
\label{fig:gate_pruning_kl}
\end{figure}

\begin{figure}[h!]
\centering
\includegraphics[width=\linewidth]{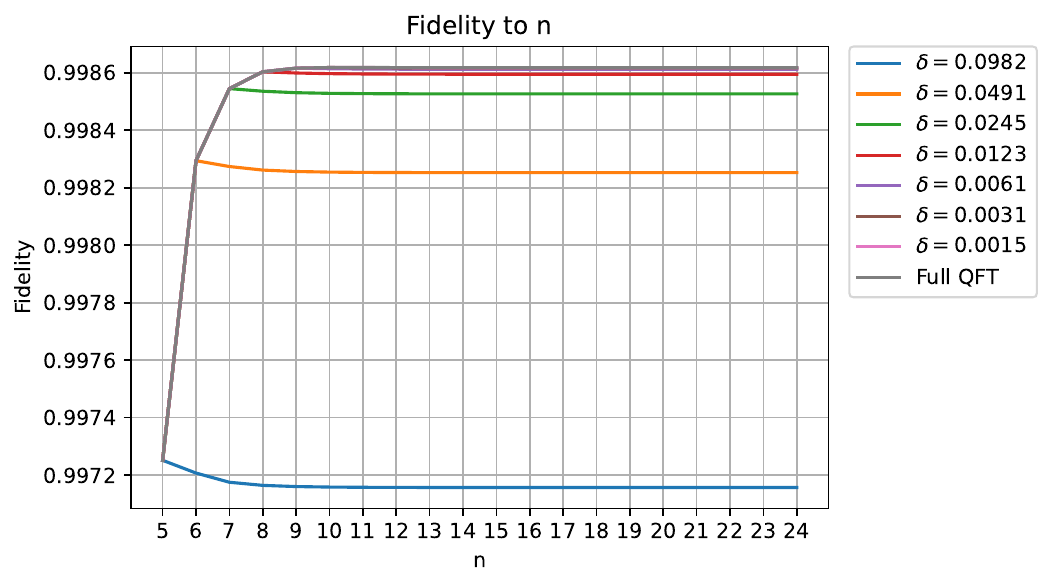}
\caption{Fidelity to \(n\) for different QFT angle threshold \(\delta\) and full QFT.}
\label{fig:gate_pruning_fidelity}
\end{figure}

The practical impact of achieving near-linear gate scaling cannot be overstated, because real devices remain limited by decoherence, crosstalk, and other noise sources that scale adversely with circuit depth \cite{Draper2000}. A near-linear approach in the QFT block can thus allow for significantly larger \(n\) (as shown in Fig.~\ref{fig:gate_pruning_scaling}) than would otherwise be feasible if we used a naive \(\mathcal{O}(n^2)\) amplitude encoding \cite{Draper2004} with marginal increase in error (as shown in Fig.~\ref{fig:gate_pruning_mse}). This is particularly relevant if the quantum algorithm that relies on the prepared Gaussian state demands minimal overhead in the state initialization phase. Fig.~\ref{fig:gate_pruning_kl} also shows that the probability distribution produced by amplitude encoding has a high KL-Divergence, while those produced by our algorithm have a decreasing trend of KL-Divergence as number of qubits increase, showing that our algorithm achives better results as the number of qubits increase due to finer discrete points. Also, one can see that when \(\delta\ge0.0123\), the change in KL-Divergence is marginal, solidifying our choice of \(\delta=0.0123\) as the threshold for gate pruning in QFT. Fig.~\ref{fig:gate_pruning_fidelity} shows a similar result as Fig.~\ref{fig:gate_pruning_kl}, as when \(\delta\ge0.0123\) and \(n\ge10\), the change in fidelity is marginal.

We also observed that for certain larger decay rates \(\lambda\), some minor deviations appeared in the tails, consistent with the notion that extremely sharp exponentials can magnify small errors from omitted phases \cite{Montanaro2015}. However, these deviations often do not significantly degrade the overall fidelity.

We then performed a small-scale hardware test on IBMQ \textit{Kyiv} \cite{Woerner2019} with results shown in Fig.~\ref{fig:ibm_result}. We used \(n\in \{4, 5, 6\}\) qubits, prepared the product state with the \(R_y\) angles described above, then applied a QFT that pruned some small angles with threshold \(\delta=0.01\), and measured the final distribution across 50,000 shots. The distribution shape was clearly reminiscent of the intended Gaussian peak for 4 and 5 qubits, although hardware noise broadened the distribution more than in the noiseless simulation. For 6 qubits, the Gaussian peak is only slightly visible due to hardware noise and error.

\begin{figure}[h!]
\centering
\includegraphics[width=\linewidth]{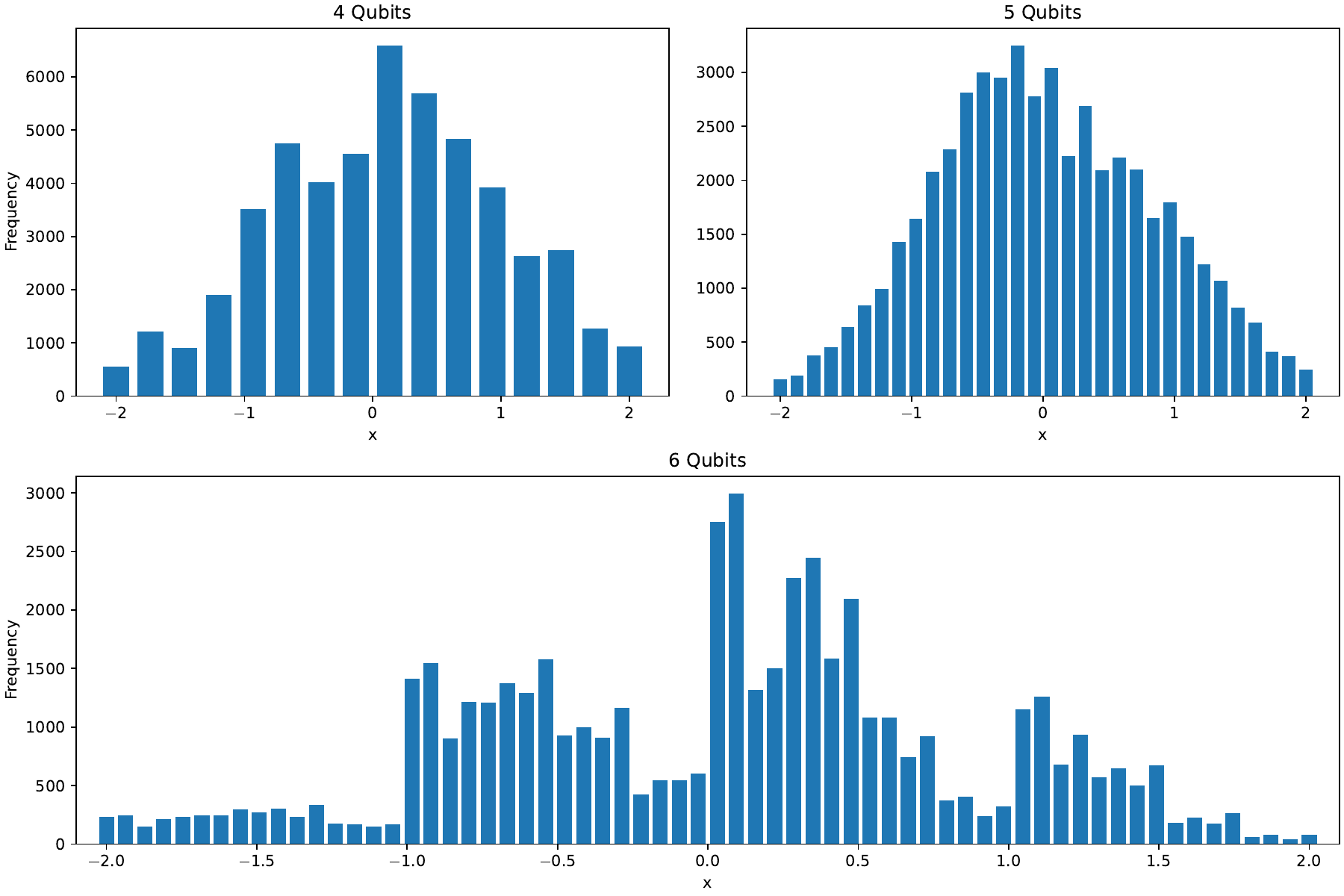}
\caption{Frequency distribution of running the circuit on IBMQ \textit{Kyiv} with \(n\in \{4, 5, 6\}\), \(\delta=0.01\), and 50,000 shots.}
\label{fig:ibm_result}
\end{figure}

This is consistent with typical gate error rates on devices in the era of near-term quantum computing \cite{Rebentrost2014}. We believe that with further error mitigation or improved hardware generations, the technique can be validated at higher qubit numbers. Regardless, the test underscores that the circuit is sufficiently compact to be run on current hardware, at least in small-scale demonstrations, thus stressing again the crucial role that efficient Gaussian state preparation can play in bridging algorithmic requirements and real-world devices \cite{Rebentrost2014}.

\section{Future Works and Conclusion}

Gaussian state preparation has a huge significance in quantum computing because many applications \cite{Stamatopoulos2020} benefit from states whose amplitudes is a smooth, continuous-variable distribution \cite{Lloyd2014}. Discrete approximations of Gaussian states are important for simulating physical systems governed by quadratic Hamiltonians, modeling sampling processes for financial applications, or encoding data in quantum machine learning tasks \cite{Rebentrost2018}.

This research has offered an algorithm that achieves this states by combining single-qubit \(R_y\) gates that induce a bitwise exponential distribution and the quantum Fourier transform that spreads these amplitudes into a near-Gaussian shape in the final basis \cite{Gily_n_2019}. By selectively omitting controlled-phase gates below a threshold \(\delta\), the circuit depth can be reduced from \(\mathcal{O}(n^2)\) to near \(\mathcal{O}(n)\), preserving high fidelity in practice.

Numerical simulations confirm that one can achieve small mean-squared errors, and preliminary hardware trials show a recognizable peak structure for small numbers of qubits \cite{Rebentrost2014}. We therefore conclude that this method is a promising candidate for scalable Gaussian state initialization on digital quantum computers, bridging theoretical requirements and practical implementation constraints.

In future research, a deeper exploration of other gaussian profiles beyond \(\lambda=1\), how to tune the angles in an adaptive manner, or how to incorporate advanced error mitigation strategies, may further boost accuracy \cite{Lloyd2014}. Additionally, integration with algorithms that explicitly require Gaussian initializations may provide immediate performance gains, for instance in quantum-enhanced data analysis or in partial differential equation solvers that rely on wavefunction-like initial states. The general principle that bitwise exponentials plus a QFT yield near-Gaussian final states should also motivate new directions in continuous-variable emulation and quantum signal processing. The implementation of this algorithm is available at the Classiq library: \verb|https://github.com/classiq/classiq-library|.

\section*{Acknowledgement}

We extend our sincere gratitude to Classiq for providing the essential platform and resources that facilitated the successful completion of this research \cite{Stamatopoulos2020}.

\bibliographystyle{IEEEtran}
\bibliography{references}

\end{document}